\documentclass[11pt,draftcls,onecolumn]{IEEEtran}

\usepackage{amsmath}
\usepackage{amsthm}
\usepackage{amssymb}
\usepackage{verbatim}
\usepackage{algorithm}
\usepackage{algorithmic}
\usepackage{graphicx}
\usepackage{paralist}
\usepackage{color}
\usepackage{cite}

\theoremstyle{plain}
\newtheorem{lemma}{Lemma}
\newtheorem{theorem}{Theorem}
\newtheorem{axiom}{Axiom}
\newtheorem{definition}{Definition}

\theoremstyle{definition}

\floatname{algorithm}{Algorithm}
\newtheorem*{remark}{Remark}
\newcommand{\argmax}{\operatornamewithlimits{argmax}}

\allowdisplaybreaks[4]

\begin{document}
%
\title{On Optimality of Myopic Policy for Restless Multi-armed Bandit Problem with Non i.i.d. Arms and Imperfect Detection}
\author{Kehao~Wang \qquad Lin~Chen \qquad Quan~Liu \qquad Khaldoun~Al~Agha
\IEEEcompsocitemizethanks{\IEEEcompsocthanksitem K.~Wang, L.~Chen and K. Al~Agha are with the Laboratoire de Recherche en Informatique (LRI), Department of Computer Science, the University of Paris-Sud XI, 91405 Orsay, France (e-mail: \{Kehao.Wang, Lin.Chen, Khaldoun.Alagha\}@lri.fr). K.~Wang and Q.~Liu is with the school of Information Engineering, Wuhan University of Technology, 430070 Hubei, China (e-mail: \{Kehao.wang, Quan.Liu\}@whut.edu.cn).}}

\IEEEcompsoctitleabstractindextext{%
\begin{abstract}
We consider the channel access problem in a multi-channel opportunistic communication system with imperfect channel sensing, where the state of each channel evolves as a non independent and identically distributed Markov process. This problem can be cast into a restless multi-armed bandit (RMAB) problem that is intractable for its exponential computation complexity. A natural alternative is to consider the easily implementable myopic policy that maximizes the immediate reward but ignores the impact of the current strategy on the future reward. In particular, we analyze a family of generic and practically important functions, termed as $g$-regular functions characterized by three axioms, and establish a set of closed-form structural conditions for the optimality of myopic policy. 
\end{abstract}

\begin{IEEEkeywords}
Restless multi-armed bandit (RMAB), myopic policy, opportunistic spectrum access (OSA), Imperfect Detection
\end{IEEEkeywords}}

\maketitle

\IEEEdisplaynotcompsoctitleabstractindextext

%
\IEEEpeerreviewmaketitle

\section{Introduction}
\label{sec:intro}

We consider the restless multi-armed bandit (RMAB) problem in the context of opportunistic multi-channel communication system in which a user has access to multiple
channels, but is limited to sense and transmit only on a subset of them at a time. The fundamental problem is how the user can exploit past observations and the knowledge of the stochastic properties of the channels to maximize its utility (e.g., expected throughput) by switching channels opportunistically.

The RMAB problem, although well defined, is proved to be PSPACH-Hard to solve~\emph{et al.} in~\cite{Papadimitriou99}, and very little result is reported on the structure of the optimal policy due to its high complexity. Recently, an alternative approach has captured extensive research attention which consists of seeking the myopic policy (also termed as greedy policy) which maximizes the expected immediate reward while ignoring the impact of the current action on the future. Zhao~\emph{et al.}~\cite{Qzhao08} established the structure of the myopic sensing policy, analyzed the performance, and partly obtained the optimality for the case of i.i.d. channels. Ahmad and Liu~\emph{et al.}~\cite{Ahmad09b} derived the optimality of the myopic sensing policy for the positively correlated i.i.d. channels when the user is limited to access one channel (i.e., $k=1$) each time, and further extended the optimality to the case of sensing multiple i.i.d. channels ($k>1$)~\cite{Ahmad09}.
In our previous work~\cite{Wang11TPS} we extended i.i.d. channels~\cite{Ahmad09b} to non i.i.d. ones, and focused on a family of generic and important utility functions, termed as \emph{regular} function, and derived closed-form conditions under which the myopic sensing policy is ensured to be optimal.
For the imperfect sensing channel model, Liu and Zhao~\emph{et al.}~\cite{Kliu10} proved the optimality of the myopic policy for the case of two channels with a particular utility function and conjectured it for arbitrary $N$. In~\cite{Wang11TCOM}, we extended the optimality of myopic policy for i.i.d. channels from the perfect sensing to the imperfect sensing, and as a consequence, derived closed-form conditions to guarantee the optimality of the myopic sensing policy for arbitrary $N$ and for regular function.


Our study presented in this paper builds upon and extends our earlier work~\cite{Wang11TPS,Wang11TCOM}. Under the assumption of imperfect channel observation, we perform an analytical study on the optimality of the myopic policy for the considered RMAB problem. The contribution of this paper, compared with~\cite{Wang11TPS,Wang11TCOM}, is two-fold:
\begin{itemize}
\item We further generalize the third axiom in~\cite{Wang11TPS} to cover a much larger class of reward functions including the logarithmic and exponential functions. The conditions of the optimality are derived in the more general with the case in~\cite{Wang11TPS} being a special subset.
\item We derive the optimality condition of the myopic policy with imperfect channel observation and non i.i.d. channels. The main technical obstacle we overcome is that in the non-perfect sensing case, the belief value of a channel depends not only on the evolution itself, but also on the observation outcome, which leads to indeterministic transition and nonlinear propagation of the belief vector.
\end{itemize}

It is worth noting that despite the vital importance, very few work has been done on the impact of imperfect observation on the performance of the myopic policy. To our knowledge,~\cite{Kliu10}
and~\cite{Wang11TCOM} are
the only analysis pertinent to our study in this paper. They both focus on i.i.d. channels, while the analysis in this paper levitates this assumption by considering the generic heterogeneous case which requires an original analysis on the optimality, as detailed later in the paper. Table 1 summarizes the related work on the myopic policy and illustrates the work presented in this paper within the context.

\begin{table}[ht]
\begin{center}
\caption{Summary of related work on myopic policy of RMAB problem}
\begin{tabular}{|l|l|l|}
\hline
& i.i.d arms & non i.i.d. arms \\
\hline
Perfect observation & \cite{Qzhao08},~\cite{Ahmad09b,Ahmad09} & \cite{Wang11TPS} \\
\hline
Imperfect observation & \cite{Kliu10}, \cite{Wang11TCOM} & this paper \\
\hline
\end{tabular}
\label{table:related_work}
\end{center}
\end{table}

The rest of the paper is organized as follows: Our model is formulated in Section~\ref{sec:model} and then the $g$-regular function is introduced in Section~\ref{sec:axiom}. Section~\ref{sec:analysis} studies the optimality of the myopic sensing policy. Finally, the paper is concluded by Section~\ref{sec:conclusion}.

\section{System Model and Problem Formulation}
\label{sec:model}

We consider the multi-channel opportunistic communication system where the user is allowed to sense only $k$ ($1\le k\le N$) of the $N$ channels at each slot $t$. The transmission probabilities of channel $i$ are $p^i_{rs},r,s=0,1$. We assume $p^i_{11}>p^i_{01}, 1\leq i \leq N$. We denote the set of channels chosen by the user at slot $t$ by ${\cal A}(t)$ where ${\cal A}(t)\subseteq {\cal N}$ and $|{\cal A}(t)|=k$. We are interested in the imperfect sensing scenario where channel sensing is subject to errors, i.e., a good channel may be sensed as bad one and vice versa. Let $\mathbf{S}(t)\triangleq[S_1(t),\cdots,S_N(t)]$ denote the channel state vector where $S_i(t)\in\{0,1\}$ is the state of channel $i$ in slot $t$ and let $\mathbf{S'}(t)\triangleq\{S'_i(t), i\in{\mathcal{A}(t)}\}$ denote the sensing outcome vector where $S_i'(t)=0$ ($1$) means that the channel $i$ is sensed bad (good) in slot $t$. Using such notation, the performance of channel state detection is characterized by two system parameters: the probability of false alarm $\epsilon_i(t)$ and the probability of miss detection $\delta_i(t)$, formally defined as follows:
\begin{eqnarray*}
  \epsilon_i(t) \triangleq \text{Pr}\{S_i'(t)=0 | S_i(t)=1\}, \\
  \delta_i(t)   \triangleq \text{Pr}\{S_i'(t)=1 | S_i(t)=0\}.
\end{eqnarray*}
In our analysis, we consider the case where $\epsilon_i(t)$ and $\delta_i(t)$ are independent w.r.t. $t$ and $i$. More specifically, we defined $\epsilon$ and $\delta$ as the system-wide false alarm rate and miss detection rate.
We assume that the user only transmits over the channel sensed to be good.


We also assume that when the receiver successfully receives a packet from a channel, it sends an acknowledgement to the transmitter over the same channel at the end of the slot. The absence of an ACK (NACK) signifies that the transmitter does not transmit over this channel or transmitted but the channel is busy in this slot. We assume that acknowledgement are received without error since acknowledgements are always transmitted over idle channels~\cite{Kliu10}.

Obviously, by sensing only $k$ out of $N$ channels, the user cannot observe the state information of the whole system. Hence, the user has to infer the channel states from its past decision and observation history so as to make its future decision. To this end, we define the \emph{channel state belief vector} (hereinafter referred to as \emph{belief vector} for briefness) $\Omega(t)\triangleq\{\omega_i(t), i\in{\cal N}\}$, where $0\le \omega_i(t)\le 1$ is the conditional probability that channel $i$ is in state good (i.e., $S_i(t)=1$) at slot $t$ given all past states, actions and observations. In order to ensure that the user and its intended receiver tune to the same channel in each slot, channel selections should be based on common observations $\{ 0~\text{(NACk)}, 1~\text{(ACK)} \}^k$ rather than the detection outcomes at the transmitter.
Due to the Markovian nature of the channel model, given the action $\mathcal{A}(t)$ and the observations $\{ ACK_i(t)\in \{0,1\}:i\in \mathcal{A}(t) \}$, the belief vector can be updated recursively using Bayes Rule as shown in~\eqref{eq:belief_update_err_asym}.
\begin{equation}
\omega_i(t+1)=
\begin{cases}
p_{11}^i, & i\in {\cal A}(t), ACK_i(t)=1 \\
\tau_i(\varphi(\omega_i(t))), & i\in {\cal A}(t), ACK_i(t)=0 \\
\tau_i(\omega_i(t)), & i\not\in {\cal A}(t)
\end{cases},
\label{eq:belief_update_err_asym}
\end{equation}
Note that the belief update under $ACK_i(t)=0$ results from the fact that the receiver cannot distinguish a failed transmission (i.e., collides with the primary user with probability $\delta(1-\omega_i(t))$) from no transmission (with probability $\epsilon\omega_i(t)+(1-\delta)(1-\omega_i(t))$)~\cite{Kliu10}. For convenience, we introduce two operators
$\varphi(\omega_i)=\frac{\epsilon \omega_{i}(t)}{\epsilon \omega_{i}(t) + 1 - \omega_{i}(t)}$ and
\begin{equation}
\tau_i(\omega_i(t))\triangleq\omega_i(t)\cdot p^i_{11}+(1-\omega_i(t))\cdot p^i_{01}.
\label{eq:tau_err_sym}
\end{equation}

\begin{remark}
We would like to emphasize that in contrast to the perfect sensing case where $\omega_i(t+1)$ is a linear function of $\omega_i(t)$ whether $i$ is sensed or not, in the imperfect sensing case, the mapping from $\omega_i(t)$ to $\omega_i(t+1)$ is no longer linear due to the sensing error (cf. the second line of equation~\eqref{eq:belief_update_err_asym}). In addition, Papadimitriou et al ~\cite{Papadimitriou99} shows that for $N$ arms, even when the active transition matrix and the passive one are deterministic transitions (e.g. either 0 or 1), computing the optimal policy is PSPACE-hard, and their proof also shows that deciding the optimal reward is non-zero is also PSPACE-hard, hence ruling out any \emph{approximation algorithm} as well. Unfortunately, the considered problem in this paper just is the case without any approximation algorithm because the belief value update of a channel depends not only on the channel evolution itself, but also on the observation outcome, i.e., $\omega_i(t+1)=\tau_i(\omega_i(t))$ for $i\notin {\cal A}(t)$ and $\omega_i(t+1)=\tau_i(\varphi(\omega_i(t)))$ for $i\in {\cal A}(t), ACK_i(t)=0$. Therefore, an original study on the optimality of the myopic sensing policy is especially required since these aforementioned differences make the analysis for the perfect sensing case no more applicable in the imperfect sensing case. It should also be noted that the perfect sensing case can be regarded as a degenerated case with $\epsilon=\delta=0$.
\end{remark}

A sensing policy $\pi$ specifies a sequence of functions $\pi=[\pi_1,\pi_2,\cdots,\pi_T]$ where $\pi_t$ maps the belief vector $\Omega(t)$ to the action (i.e., the set of channels to sense) ${\mathcal A}(t)$ in each slot $t$: $\pi_t: \ \Omega(t)\rightarrow{\mathcal A}(t), |{\mathcal A}(t)|=k$.

Given the imperfect sensing context, we are interested in the user's optimization problem to find the optimal sensing policy $\pi^*$ that maximizes the expected total discounted reward over a finite horizon:
\begin{equation}
\pi^*=\argmax_{\pi} \mathbb{E}\left.\left[\sum^{T}_{t=1} \beta^{t-1} R(\pi_t(\Omega(t)))\right|\Omega(1)\right]
\label{eq:pb_formulation_err_sym}
\end{equation}
where $R(\pi_t(\Omega(t)))$ is the reward collected in slot $t$ under the sensing policy $\pi_t$ with the initial belief vector $\Omega(1)$\footnote{If no information on the initial system state is available, each entry of $\Omega(1)$ can be set to the stationary distribution $\omega^{i}_0=\frac{p^{i}_{01}}{1+p^{i}_{01}-p^{i}_{11}}$, $1\leq i \leq N$.}, $0\le \beta\le 1$ is the discounted factor characterizing the feature that the future rewards are less valuable than the immediate reward. By treating the belief value of each channel as the state of each arm of a bandit, the user's optimization problem can be cast into a restless multi-armed bandit problem.

In this paper, we focus on the myopic sensing policy which is easy to compute and implement that maximizes the immediate reward, formally defined as follows:

\begin{definition}[Myopic Sensing Policy]
Let $F(\Omega_A(t))\triangleq \mathbb{E}[R(\pi_t(\Omega(t)))]$ denote the expected immediate reward obtained in slot $t$ under the sensing policy $\pi_t$, the myopic sensing policy $\mathcal{\widetilde{A}}(t)$, consists of sensing the $k$ channels that maximizes $F(\Omega_A(t))$, i.e., $\mathcal{\widetilde{A}}(t)\triangleq \argmax_{\mathcal{A}(t)\subseteq \mathcal{N}} F(\Omega_A(t))$.
\end{definition}

In the sequel analysis, we establish closed-form conditions under which the myopic sensing policy is guaranteed to be optimal. 
Before ending this section, we state some structural properties of $\tau_i(\omega_i(t))$ and $\varphi(\omega_i(t))$ that are useful in the subsequent proofs.
\begin{lemma}
For any positively correlated channel $i$ (i.e., $p_{01}^i<p_{11}^i$), the following structural properties of $\tau_i(\omega_i(t))$ hold:
\begin{itemize}
\item $\tau_i(\omega_i(t))$ is monotonically increasing in $\omega_i(t)$;
\item $p_{01}^i\le \tau_i(\omega_i(t))\le p_{11}^i$, $\forall \ 0\le \omega_i(t)\le 1$.
\end{itemize}
\label{lemma:property_tau_err_sym}
\end{lemma}
\begin{proof}
Noticing that $\tau_i(\omega_i(t))$ can be written as $\tau_i(\omega_i(t))=(p_{11}^i-p_{01}^i)\omega_i(t)+p_{01}^i$, Lemma~\ref{lemma:property_tau_err_sym} holds straightforwardly.
\end{proof}

\begin{lemma}
$\varphi(\omega_i(t))$ monotonically increases with $\omega_i(t)$ when $0 \le \epsilon <1$.
\label{lemma:property_phi_err_sym}
\end{lemma}
\begin{proof}
Noticing that $\varphi(\omega_i)=\frac{\epsilon \omega_{i}(t)}{\epsilon \omega_{i}(t) + 1 - \omega_{i}(t)}$, Lemma~\ref{lemma:property_phi_err_sym} follows straightforwardly.
\end{proof}

\section{Axioms}
\label{sec:axiom}


This section defines three axioms characterizing a family of generic and practically important functions referred to as \emph{$g$-regular} functions, which serve as a basis for the further analysis on the structure and the optimality of the myopic sensing policy. Without ambiguity, we drop the time index of $\omega_i(t)$, and abuse $\omega_i(t)$ and $\omega_i$ alternatively.

\begin{axiom}[Symmetry~\cite{Wang11TPS}]
\label{axiom:symmetry}
A function $f(\Omega_A): [0,1]^k\rightarrow\mathbb{R}$ is symmetrical if for any two distinct channels $i$ and $j$, it holds that
\begin{equation*}
f(\omega_{1}, \cdots, \omega_{i}, \cdots, \omega_{j}, \cdots, \omega_{k})=f(\omega_{1}, \cdots, \omega_{j}, \cdots, \omega_{i}, \cdots, \omega_{k}).
\end{equation*}
\end{axiom}

\begin{axiom}[Monotonicity~\cite{Wang11TPS}]
\label{axiom:monotonicity}
A function $f(\Omega_A): [0,1]^k\rightarrow\mathbb{R}$ is monotonically increasing if it is monotonically increasing in each variable $\omega_i$, i.e.,
\begin{equation*}
\omega'_i>\omega_i \Longrightarrow f(\omega_1, \cdots, \omega'_i, \cdots, \omega_k)>f(\omega_1, \cdots, \omega_i, \cdots, \omega_k), \quad \forall i\le k.
\end{equation*}
\end{axiom}

The above axioms are the intuitive with Axiom~\ref{axiom:symmetry} stating that once the sensing set $\cal A$ is given, the sensing order will not change the final reward under a symmetrical function $f$. The following axiom, however, significantly extends the axiom of decomposability in~\cite{Wang11TPS} so as to cover a much larger range of utility functions.

\begin{axiom}[$g$-Decomposability]
\label{axiom:decomposability}
A function $f(\Omega_A): [0,1]^k\rightarrow\mathbb{R}$ is decomposable if there exists a continuous and increasing function $g:[0,1]\rightarrow[0,\infty)$ and a constant $c$ such that for any $i\le k$ it holds that
\begin{multline*}
f(\omega_{1}, \cdots, \omega_{i-1}, \omega_i, \omega_{i+1}, \cdots, \omega_k) = c\cdot g(\omega_i) f(\omega_{1}, \cdots, \omega_{i-1}, 1, \omega_{i+1}, \cdots, \omega_k)  \\
        + c\cdot (1-g(\omega_i))f(\omega_{1}, \cdots, \omega_{i-1}, 0, \omega_{i+1}, \cdots, \omega_k).
\end{multline*}
\end{axiom}

Axiom~\ref{axiom:decomposability} on the $g$-decomposability states that $f({\Omega_A})$ can always be decomposed into two terms by introducing the function $g$ and replacing $\omega_i$ by $0$ and $1$, respectively. It is insightful to note Axiom of $g$-decomposability significantly extends Axiom of decomposability in~\cite{Wang11TPS} by covering a much larger range of utility functions which cannot be covered by latter, particularly the logarithmic function (e.g., $f(\Omega_A)=\sum_{i=1}^k \log_{a}(1+\omega_i)$ ($a>1$), where $c=\frac{1}{\log_{2}a}$, $g(\omega_i)=\log_{2}(1+\omega_i)$ ) and the power function (e.g., $f(\Omega_A)=\sum_{i=1}^k \omega_i^a, a>0$, where $c=1$, $g(\omega_i)=\omega_i^a$) that are widely used in engineering problems. By setting $g(\omega_i)=\omega_i$ and $c=1$, Axiom~\ref{axiom:decomposability} degenerates to the Axiom of decomposability in~\cite{Wang11TPS}.

In the following, we use the above axioms to characterize a family of generic functions, referred to as \emph{$g$-regular} functions, defined as follows.

\begin{definition}[$g$-Regular Function]
A function is called $g$-regular if it satisfies all the three axioms.
\end{definition}

If the expected reward function $F$ is $g$-regular, the myopic sensing policy, defined in Definition 1, consists of sensing the $k$ channels with the largest belief values. In case of tie, we can sort the channels in tie in the descending order of $\omega_i(t+1)$ calculated in~\eqref{eq:belief_update_err_asym}. The argument is that larger $\omega_i(t+1)$ leads to larger expected payoff in next slot $t+1$. If the tie persists, then the channels are sorted by their indexes.

\section{Analysis on Optimality of Myopic Sensing Policy under Imperfect Sensing}
\label{sec:analysis}

In this section, we establish the closed-form conditions under which the myopic sensing policy achieves the system optimum under imperfect sensing. To this end, we set up by defining a pseudo value function and studying its structural properties which are then used to establish the main result on the optimality.

\subsection{Pseudo Value Function}
\label{sec:property_value_function}

Armed with the three axioms, this section first defines the \emph{pseudo value function} in the imperfect sensing case and then derives several fundamental properties of it, which are crucial in the study on the optimality of the myopic sensing policy. We start by giving the formal definition of the pseudo value function in the recursive form.

\begin{definition}[Pseudo Value Function]
The pseudo value function, denoted as $W_t(\Omega_A(t))$ ($1\leq t\leq T$, $t+1 \leq r\leq T$) is recursively defined as follows:
\begin{equation}
\begin{cases}
W_T(\Omega(T))= F(\Omega_{\widetilde{A}}(T)); \\
W_r(\Omega(r))= F(\Omega_{\widetilde{A}}(r))+  \beta \sum_{{\mathcal E}\subseteq{\mathcal{\widetilde{A}}(r)}} Pr(\mathcal{\widetilde{A}}(r), {\mathcal E})W_{r+1}(\Omega_{{\mathcal E}}(r+1)); \\
W_t(\Omega_A(t))= F(\Omega_{A}(t))+  \beta \underbrace{\sum_{{\mathcal E}\subseteq{\mathcal{A}(t)}} Pr(\mathcal{A}(t), {\mathcal E})W_{t+1}(\Omega_{{\mathcal E}}(t+1))}_{\Gamma(\Omega_A(t))},
\end{cases}
\label{eq:w_t}
\end{equation}
where $\Omega_{{\mathcal E}}(t+1)$ and $\Omega_{{\mathcal E}}(r+1)$ are generated by $\langle\Omega(t),\mathcal{A}(t),\mathcal{E}\rangle$ and $\langle\Omega(r),\mathcal{\widetilde{A}}(r),\mathcal{E}\rangle$, respectively, according to~\eqref{eq:belief_update_err_asym}, and $\displaystyle Pr({\cal M},{\mathcal E})\triangleq\prod_{i\in{\mathcal E}}(1-\epsilon)\omega_{i}(t)\prod_{j\in{\cal M}\setminus{\mathcal E}}[1-(1-\epsilon)\omega_{j}(t)]$.
\end{definition}

The pseudo value function gives the expected discounted accumulated reward of the following sensing policy: in slot $t$ sense the channels in $\mathcal{A}(t)$ and then sense the channels in $\mathcal{\widetilde{A}}(r) ~(t+1 \leq r\leq T)$ (i.e., adopt the myopic policy from slot $t+1$ to $T$). If $\mathcal{A}(t)=\mathcal{\widetilde{A}}(t)$, then the above sensing policy is the myopic sensing policy with $W_t(\Omega_A(t))$ being the total reward from slot $t$ to $T$.

\begin{lemma}
\label{lemma:symmetry_asym_err}
If the expected reward function $F(\Omega_A)$ is $g$-regular, the correspondent pseudo value function $W_t(\Omega_A(t))$ is symmetrical about $\omega_i,\omega_j$ where $i, j\in{\cal A} \text{ or }i, j \notin{\cal A}$ for all $t=1, 2, \cdots, T$.
\end{lemma}
\begin{proof}
The lemma can be easily shown by backward induction noticing that $F(\Omega_A)$ is symmetrical about $\omega_i,\omega_j$, and $(\omega_{1}, \cdots, \omega_i, \cdots, \omega_j, \cdots, \omega_N)$ and $(\omega_{1}, \cdots, \omega_j, \cdots, \omega_i, \cdots, \omega_N)$ generate the same belief vector $\Omega(t+1)$ no matter whether $i, j\in{\cal A} \text{ or }i, j \notin{\cal A}$, combined with the fact that the myopic policy is adopted from slot $t+1$ to $T$ by~\eqref{eq:w_t}, we conclude $W_{t+1}(\Omega_{{\mathcal E}}(t+1))$ is symmetrical about $\omega_i,\omega_j$. Thus the lemma holds.
\end{proof}



\subsection{Myopic Sensing Policy: Condition of Optimality}
\label{sec:optimality_myopic}
In this subsection, we study the optimality of the myopic sensing policy. For the convenience of discussion, we firstly state some notation before presenting the analysis.

\begin{itemize}
\item $\displaystyle p_{11}^{max}\triangleq\max_{i\in{\cal N}} \{p_{11}^i\}$, $\displaystyle p_{01}^{min}\triangleq\max_{i\in{\cal N}} \{p_{01}^i\}$;
\item $\displaystyle \delta_p^{max}\triangleq\max_{i\in{\cal N}} \{p_{11}^i-p_{01}^i\}$, $\displaystyle \delta_p^{min}\triangleq\min_{i\in{\cal N}} \{p_{11}^i-p_{01}^i\}$;
\item $\displaystyle g'_{min}\triangleq\min_{ p^{min}_{01} \leq \omega \leq p^{max}_{11}} \Big\{ \frac{dg(\omega)}{d\omega} \Big\}$, $\displaystyle g'_{max}\triangleq\max_{p^{min}_{01} \leq \omega \leq p^{max}_{11}} \Big\{ \frac{dg(\omega)}{d\omega} \Big\}$;
\item Let $\omega_{-i}\triangleq\{\omega_j:j\in{\cal A},j\ne i\}$ denote the believe vector except $\omega_i$, and
\begin{eqnarray*}
\begin{cases}
\displaystyle \Delta_{max}\triangleq \max_{ \omega_{-i}\in[0,1]^{N-1}} \ \{F(1, \omega_{-i})-F(0, \omega_{-i})\}, \\
\displaystyle \Delta_{min}\triangleq \min_{ \omega_{-i}\in[0,1]^{N-1}} \ \{F(1, \omega_{-i})-F(0, \omega_{-i})\}.
\end{cases}
\end{eqnarray*}
\end{itemize}

We start by showing the following important lemma (Lemma~\ref{lemma:decomposability_asym_err}) and then establish the sufficient condition under which the optimality of the myopic sensing policy is ensured. In Lemma~\ref{lemma:decomposability_asym_err}, we consider $\Omega_l=[\omega_1,\cdots,\omega_l,\cdots,\omega_N]$ and $\Omega'_l=[\omega_1,\cdots,\omega'_l,\cdots,\omega_N]$ which differ only in one element $\omega'_l \ge \omega_l$. Let $\mathcal{A'}$ and $\mathcal{A}$ denote the largest $k$ elements in $\Omega'_l$ and $\Omega_l$, respectively\footnote{The tie, if exists, is resolved in the way as stated in remark after Definition 3}, Lemma~\ref{lemma:decomposability_asym_err} gives the upper and lower bounds of $W_t(\Omega_{A'})-W_t(\Omega_A)$.

\begin{lemma}
\label{lemma:decomposability_asym_err}
If the expected reward function $F(\Omega_A)$ is $g$-regular, $\forall l\in{\cal N}$, $\omega_l\le \omega_l'$ and $1\le t\le T$, we have
\begin{enumerate}
\item if $l\in{\cal A}'$ and $l\in{\cal A}$, then
    \begin{equation*}
    c\cdot(\omega_l'-\omega_l)g'_{min}\Delta_{min} \le W_t(\Omega_{A'})-W_t(\Omega_{A})
    \le c\cdot(\omega_l'-\omega_l)g'_{max}\Delta_{max}\sum_{i=0}^{T-t}\beta^i(\delta_p^{max})^i;
    \end{equation*}
\item  if $l\notin { {\cal A}' }$ and $l\notin{ {\cal A}}$, then $\displaystyle 0 \le W_t(\Omega_{A'})-W_t(\Omega_{A}) \le c\cdot(\omega_l'-\omega_l)g'_{max}\Delta_{max}\sum_{i=1}^{T-t}\beta^i(\delta_p^{max})^i$;
\item if $l\in {\cal A}'$ and $l\notin { {\cal A}}$, then $\displaystyle 0 \le W_t(\Omega_{A'})-W_t(\Omega_{A}) \le c\cdot(\omega_l'-\omega_l)g'_{max}\Delta_{max}\sum_{i=0}^{T-t}\beta^i(\delta_p^{max})^i$.
\end{enumerate}
\end{lemma}
\begin{proof}
The proof is given in the Appendix~\ref{appendix:decomposability_asym_err}.
\end{proof}

\begin{remark}
It can be noted that the case $l\notin {\cal A}'$ and $l\in { {\cal A}}$ is impossible to exist according to the definition of the myopic sensing policy.
\end{remark}

In the following lemma, we consider $W_t(\Omega_{A_l})$ and $W_t(\Omega_{A_m})$ where ${\cal A}_l$ and ${\cal A}_m$ differ in one element ($l\in{\cal A}_l$ and $m\in{\cal A}_m$ and $\omega_l>\omega_m$). Lemma~\ref{lemma:bound_myopic_sensing_err} establishes the sufficient condition under which $W_t(\Omega_{A_l})\geq W_t(\Omega_{A_m})$ when $F$ is $g$-regular.

\begin{lemma}
If $F(\Omega_A)$ is $g$-regular and $\displaystyle \frac{g'_{min}\Delta_{min}}{g'_{max}\Delta_{max}} \geq \sum_{i=1}^{T-1}\beta^i(\delta_p^{max})^i$, then $W_t(\Omega_{A_l}) \geq W_t(\Omega_{A_m})$ holds for $1 \leq t \leq T$.
\label{lemma:bound_myopic_sensing_err}
\end{lemma}

\begin{proof}
Let $\Omega'$ denote the set of channel belief values with $\omega_l'=\omega_m$ and $\omega_i'=\omega_i$ for $\forall i\ne l$, apply Lemma~\ref{lemma:decomposability_asym_err}, we have
\begin{align*}
&W_t(\Omega_{A_l})-W_t(\Omega_{A_m})=[W_t(\Omega_{A_l})-W_t(\Omega')]-[W_t(\Omega_{A_m})-W_t(\Omega')] \\
\geq & c\cdot(\omega_l-\omega_m)g'_{min} \Delta_{min}-c\cdot(\omega_l-\omega_m)g'_{max}\Delta_{max}\sum_{i=1}^{T-t}\beta^i(\delta_p^{max})^i \\
\ge& c\cdot(\omega_l-\omega_m)g'_{max}\Delta_{max}\cdot \left[\frac{g'_{min}}{g'_{max}}\cdot\frac{\Delta_{min}}{\Delta_{max}}-\sum_{i=1}^{T-1}\beta^i(\delta_p^{max})^i\right] \geq 0
\label{eq:aux_proof1}
\end{align*}
if the conditions in the lemma hold.
\end{proof}

The following theorem studies the optimality of the myopic sensing policy under imperfect sensing. The proof is similar to that of Theorem~1 in~\cite{Wang11TPS} and is thus omitted here.

\begin{theorem}
\label{theorem:optimal_condition_pos_case}
The myopic sensing policy is optimal if the following two conditions hold: (1) the expected slot reward function $F$ is $g$-regular; (2) $\displaystyle \frac{g'_{min}\Delta_{min}}{g'_{max}\Delta_{max}} \geq \sum_{i=1}^{T-1}\beta^i(\delta_p^{max})^i$.
\end{theorem}

Theorem~\ref{theorem:optimal_condition_pos_case} generalizes the results with perfect sensing (Theorem~1 in our previous work~\cite{Wang11TPS}) in two aspects. First, with the more generic axiom on the decomposability of the expected slot reward function, the result can now cover a much larger class of reward functions including the logarithmic and power functions which are widely encountered in practical scenarios. Secondly, Theorem~\ref{theorem:optimal_condition_pos_case} also generalizes the optimality of myopic sensing policy to cover the imperfect sensing case. 

The following theorem further establishes the optimality conditions in asymptotic case $T\rightarrow \infty$. The proof follows straightforwardly from Theorem~\ref{theorem:optimal_condition_pos_case} by noticing that $\sum_{i=1}^{\infty}x^i=x/(1-x)$ for any $x\in(0,1)$.
\begin{theorem}
\label{theorem:optimal_condition_pos_case_inf}
In the infinite horizon case $T\rightarrow \infty$, the myopic sensing policy is optimal if the following conditions hold: (1) the expected slot reward function $F$ is $g$-regular; (2) $\displaystyle \beta \leq \frac{g'_{min}\Delta_{min}}{(g'_{min}\Delta_{min}+g'_{max}\Delta_{max})\delta_p^{max}}$.
\end{theorem}

\subsection{Discussion}


We consider the channel access problem where a user is limited to sense $k$ of $N$ i.i.d. channels and gets one unit of reward if the sensed channel is in the good state, i.e., the utility function can be formulated as $F(\Omega_A)=\sum_{i\in A}[(1-\epsilon)\omega_i]$. To that end, we apply Theorem 1 of~\cite{Wang11TCOM} and have $\Delta_{min}=\Delta_{max}=1-\epsilon$. We can then verify that when $\epsilon <\frac{p_{01}(1-p_{11})}{P_{11}(1-p_{01})}$, it holds that $\frac{\Delta_{min}}{\Delta_{max}\big[(1-\epsilon)(1-p_{01}) + \frac{\epsilon(p_{11}-p_{01})}{1-(1-\epsilon)(p_{11}-p_{01})}\big]}>1$. Therefore, when the condition 1 and 2 of Theorem 1 in~\cite{Wang11TCOM} hold, the myopic sensing policy is always optimal for any $\beta$, which significantly extends the results obtained in~\cite{Kliu10}.
Regarding the similar scenario with non i.i.d. channels, we have $c=1$, $g(\omega)=\omega$ and $\Delta_{min}=\Delta_{max}=1-\epsilon$, and furthermore know that the myopic policy is optimal for any $\beta$ and $\epsilon$ if $\delta_p^{max} \leq 0.5$ according to Theorem~\ref{theorem:optimal_condition_pos_case_inf}. Compared to the optimal conditions~\cite{Wang11TCOM} for i.i.d. channels, although all focusing on the optimality of the myopic policy, the closed-form conditions of optimality derived in this paper are much stricter with respect to the transmission probabilities ($\delta_p^{max} \leq 0.5$ in our paper) but much looser in false alarm rate ($\epsilon <\frac{p_{01}(1-p_{11})}{P_{11}(1-p_{01})}$ in~\cite{Wang11TCOM}). The stricter constraint on the transmission probabilities is due to the proposed method itself which sacrifices part of the optimality to cover the case of non i.i.d. channels, while the looser constraint on the sensing error comes from the fact that all the channels are only discriminated as sensed channels or non-sensed channels at each slot under which the sensing error can be absorbed without any constraint.

\section{Conclusion}
\label{sec:conclusion}

We have investigated the optimality of the myopic policy in the RMAB problem with imperfect sensing, 
and developed three axioms characterizing a family of generic and practically important functions which we refer to as $g$-regular functions. By performing a mathematical analysis based on the developed axioms, we have characterized the closed-form conditions under which the optimality of the myopic policy is guaranteed.
As future work, a natural direction we are pursuing is to investigate the RMAB problem with multiple players with potentially conflicts among them and to study the structure and the optimality of the myopic policy in that context.

\appendices

\section{Proof of Lemma~\ref{lemma:decomposability_asym_err}}
\label{appendix:decomposability_asym_err}

We prove the lemma by backward induction.

For slot $T$, noticing that $W_T(\Omega_A)=F(\Omega_A)$ and that $g'_{min}\le \frac{g(\omega)-g(\omega')}{\omega-\omega'}\le g'_{max}$ for any $p_{01}^{min}\le \omega'\le \omega\le p_{11}^{max}$, we have
\begin{enumerate}
\item For $l\in{\cal A}', $ $l\in{\cal A}$, it holds that
\begin{equation*}
c\cdot (\omega_l'-\omega_l)g'_{min}\Delta_{min} \le W_T(\Omega_{A'}) - W_T(\Omega_A) \le c\cdot[g(\omega_l')-g(\omega_l)]\Delta_{max}  \le c\cdot(\omega_l'-\omega_l)g'_{max}\Delta_{max};
\end{equation*}
\item For $l\notin{{\cal A}'}$, it holds that $l\notin{ {\cal A}}$, $W_T(\Omega_{A'}) - W_T(\Omega_A) = 0$;
\item For $l\in {\cal A}'$, $l\notin{{\cal A}}$, it exists at least one channel $m$ such that $\omega'_l\geq \omega_m \geq \omega_l$. It then holds that
\begin{multline*}
    0 \leq  c\cdot(\omega_l'-\omega_l)g'_{min}\Delta_{min} \leq W_T(\Omega_{A'}) - W_T(\Omega_A) \leq c\cdot[g(\omega_l')-g(\omega_m)]\Delta_{max} \\  \leq  c\cdot[g(\omega_l')-g(\omega_l)]\Delta_{max} \le c\cdot(\omega_l'-\omega_l)g'_{max}\Delta_{max};
\end{multline*}
\end{enumerate}
Therefore, Lemma~\ref{lemma:decomposability_asym_err} holds for slot $T$.

Assume that Lemma~\ref{lemma:decomposability_asym_err} holds for $T, \cdots, t+1$. We now prove the lemma for slot $t$.

\textbf{We first prove the first case: $l\in{{\cal A}'}$ and $l\in{\cal A}$}.
By rewriting $\Gamma(\Omega_A(t))$ in~\eqref{eq:w_t} and developing $\omega_l(t+1)$ in $\Omega(t+1)$ , we have:
\begin{eqnarray}
\label{eq:Gamma_1}
&&\Gamma(\Omega_{A'})=(1-\epsilon)\omega_l'(t) \Gamma(\Omega_{A'}^1) + (1-(1-\epsilon)\omega_l'(t)) \Gamma(\Omega_{A'}^{\varphi(\omega'_l)}) \\
&&\Gamma(\Omega_{A})=(1-\epsilon)\omega_l(t) \Gamma(\Omega_{A}^1) + (1-(1-\epsilon)\omega_l(t)) \Gamma(\Omega_{A}^{\varphi(\omega_l)})
\label{eq:Gamma_2}
\end{eqnarray}
where, $\Omega_{A'}^1$ and $\Omega_{A'}^{\varphi(\omega'_l)}$ denote $\Omega_{A'}$ with $\omega'_l(t)=1$ and $\varphi(\omega'_l)$ , respectively, while $\Omega_{A}^1$ and $\Omega_{A}^{\varphi(\omega_l)}$ denote $\Omega_{A}$ with $\omega_l(t)=1$ and $\varphi(\omega_l)$, respectively.

Noticing $\Omega_{A'}^1=\Omega_{A}^1$, we have
\begin{align*}
\Gamma(\Omega_{A'})-\Gamma(\Omega_{A})=&(1-\epsilon)(\omega_l'(t)-\omega_l(t))[\Gamma(\Omega_{A'}^1)-\Gamma(\Omega_{A'}^{\varphi(\omega'_l)})]  \\
& +(1-(1-\epsilon)\omega_l(t)) [\Gamma(\Omega_{A'}^{\varphi(\omega'_l)})-\Gamma(\Omega_{A}^{\varphi(\omega_l)})]
\end{align*}

Considering the whole realization of the belief vector, we further have
\begin{align}
\Gamma(\Omega_{A'}(t)) - \Gamma(\Omega_{A}(t))
=& \sum_{{\mathcal E}\subseteq{\cal A}(t)\setminus{\{l\}}} \prod_{i\in{\mathcal E}}(1-\epsilon)\omega_i(t)\prod_{j\in{\cal A}(t)\setminus{\mathcal E}\setminus{\{l\}}} [1-(1-\epsilon)\omega_j(t)] \cdot \nonumber \\
&\Big\{ (1-\epsilon)(\omega_l'(t)-\omega_l(t))[W_{t+1}(\Omega_{l=1}(t+1))-W_{t+1}(\Omega_{l=\varphi(\omega'_l)}(t+1))] \nonumber \\
&+ (1-(1-\epsilon)\omega_l(t))[W_{t+1}(\Omega_{l=\varphi(\omega'_{l})}(t+1))-W_{t+1}(\Omega_{l=\varphi(\omega_{l})}(t+1))] \Big\}
\label{eq:ih1_gamma_dif}
\end{align}
where, $\Omega_{l=a}(t+1)$ ($a\in\{1,\varphi(\omega'_l),\varphi(\omega_l)\}$) denotes the belief vector at slot $t+1$ under $\Omega(t)$ with $\omega_l(t+1)=\tau_l(a)$.

Next, we derive the bound of $W_{t+1}(\Omega_{l=1}(t+1))-W_{t+1}(\Omega_{l=\varphi(\omega'_l)}(t+1))$ through three cases\footnote{It can be noted that the case $l\notin \mathcal{A'}(t+1)$ and $l\in \mathcal{A}(t+1)$ is impossible.}:

\begin{itemize}
\item Case 1: if $l\in \mathcal{A'}(t+1)$ and $l\in \mathcal{A}(t+1)$, according to the induction hypothesis, we have
    \begin{align*}
    0 \leq  c\cdot (p_{11}^l-\tau_l(\varphi(\omega'_l)))g'_{min}\Delta_{min} \leq &W_{t+1}(\Omega_{l=1}(t+1))-W_{t+1}(\Omega_{l=\varphi(\omega'_l)}(t+1)) \\
    \leq  &c\cdot (p_{11}^l-\tau_l(\varphi(\omega'_l)))g'_{max}\Delta_{max}\sum_{i=0}^{T-t-1}\beta^i(\delta_p^{max})^i
    \end{align*}
\item Case 2: if $l\notin \mathcal{A'}(t+1)$ and $l\notin \mathcal{A}(t+1)$, according to the induction hypothesis, we have
    \begin{multline*}
    0 \leq W_{t+1}(\Omega_{l=1}(t+1))-W_{t+1}(\Omega_{l=\varphi(\omega'_l)}(t+1))
        \leq c\cdot(p_{11}^l-\tau_l(\varphi(\omega'_l))) g'_{max}\Delta_{max}\sum_{i=1}^{T-t-1}\beta^i(\delta_p^{max})^i
    \end{multline*}
\item Case 3: if $l\in \mathcal{A'}(t+1)$ and $l\notin \mathcal{A}(t+1)$, according to the induction hypothesis, we have
    \begin{multline*}
    0\leq  W_{t+1}(\Omega_{l=1}(t+1))-W_{t+1}(\Omega_{l=\varphi(\omega'_l)}(t+1))
    \leq c\cdot(p_{11}^l-\tau_l(\varphi(\omega'_l))) g'_{max}\Delta_{max}\sum_{i=0}^{T-t-1}\beta^i(\delta_p^{max})^i
    \end{multline*}
\end{itemize}

Combining the three cases, we obtain
\begin{align}
   0 &\leq W_{t+1}(\Omega_{l=1}(t+1))-W_{t+1}(\Omega_{l=\varphi(\omega'_l)}(t+1)) \nonumber \\
     &\leq c\cdot(p_{11}^l-\tau_l(\varphi(\omega'_l))) g'_{max}\Delta_{max}\sum_{i=0}^{T-t-1}\beta^i(\delta_p^{max})^i \nonumber\\
    &= c\cdot\Big[1-\frac{\epsilon \omega'_l}{1-(1-\epsilon)\omega'_l}\Big](p_{11}^l-p_{01}^l) g'_{max}\Delta_{max}\sum_{i=0}^{T-t-1}\beta^i(\delta_p^{max})^i.
\label{eq:ih1_gamma_dif_1}
\end{align}

According to Lemma~\ref{lemma:property_tau_err_sym} and~\ref{lemma:property_phi_err_sym}, we have $\tau_l(\varphi(\omega'_l)) \geq \tau_l(\varphi(\omega_l))$ when $\omega'_l \geq \omega_l$. Thus we have the bounds of $W_{t+1}(\Omega_{l=\varphi(\omega'_{l})}(t+1))-W_{t+1}(\Omega_{l=\varphi(\omega_{l})}(t+1))$ by the similar induction as follows:
\begin{align}
   0&\leq W_{t+1}(\Omega_{l=\varphi(\omega'_{l})}(t+1))-W_{t+1}(\Omega_{l=\varphi(\omega_{l})}(t+1)) \nonumber\\
   &\leq c\cdot(\tau_l(\varphi(\omega'_l))-\tau_l(\varphi(\omega_l))) g'_{max}\Delta_{max}\sum_{i=0}^{T-t-1}\beta^i(\delta_p^{max})^i  \nonumber \\
   &= c\cdot\frac{\epsilon(\omega'_l-\omega_l)}{[1-(1-\epsilon)\omega'_l][1-(1-\epsilon)\omega_l]} (p_{11}^l-p_{01}^l) g'_{max}\Delta_{max}\sum_{i=0}^{T-t-1}\beta^i(\delta_p^{max})^i.
\label{eq:ih1_gamma_dif_2}
\end{align}

Combining~\eqref{eq:ih1_gamma_dif}, \eqref{eq:ih1_gamma_dif_1} and~\eqref{eq:ih1_gamma_dif_2} and recalling $p_{11}^l-p_{01}^l\leq \delta_p^{max}$, we have
\begin{equation*}
    0 \leq \Gamma(\Omega_{A'}(t)) - \Gamma(\Omega_{A}(t))
    \leq  c\cdot(\omega_l'-\omega_l) \delta_p^{max} g'_{max} \Delta_{max}\sum_{i=0}^{T-t-1}\beta^i(\delta_p^{max})^i.
\end{equation*}

Since $\Gamma(\Omega_{A'}(t))-\Gamma(\Omega_{A}(t))\geq 0$ and $$ c\cdot(\omega_l'-\omega_l)g'_{min}\Delta_{min} \leq  F(\Omega_{A'}(t)) - F(\Omega_{A}(t)) \leq c\cdot(\omega_l'-\omega_l)g'_{max}\Delta_{max},$$ we have
\begin{align*}
 c\cdot(\omega_l'-\omega_l)g'_{min}\Delta_{min} & \leq  W_t(\Omega_{A'}(t)) - W_t(\Omega_{A}(t)) \\
 & =F(\Omega_{A'}(t)) - F(\Omega_{A}(t)) + \beta (\Gamma(\Omega_{A'}(t))-\Gamma(\Omega_{A}(t))) \\
 &\le c\cdot(\omega_l'-\omega_l)g'_{max}\Delta_{max}
  + \beta \cdot c\cdot(\omega_l'-\omega_l) g'_{max} \delta_p^{max} \Delta_{max}\sum_{i=0}^{T-t-1}\beta^i(\delta_p^{max})^i \\
 &= c\cdot(\omega_l'-\omega_l)g'_{max}\Delta_{max}\sum_{i=0}^{T-t}\beta^i(\delta_p^{max})^i.
\end{align*}
We thus complete the proof of the first part ($l\in{{\cal A}'}$ and $l\in{\cal A}$) of Lemma~\ref{lemma:symmetry_asym_err}.

Secondly, \textbf{we prove the second case $l\notin{{\cal A}'}$ and $l\notin{\cal A}$}. To this end, we have:
\begin{eqnarray*}
\begin{cases}
\displaystyle\Gamma(\Omega_{A}(t)) = \sum_{{\mathcal E}\subseteq{\cal A}(t)} \prod_{i\in{\mathcal E}}(1-\epsilon)\omega_i(t)\prod_{j\in{\cal A}(t)\setminus{\mathcal E}}[1-(1-\epsilon)\omega_j(t) ]W_{t+1}(\Omega_l(t+1)) \\
\displaystyle\Gamma(\Omega_{A'}(t)) = \sum_{{\mathcal E}\subseteq{\cal A}(t)} \prod_{i\in{\mathcal E}}(1-\epsilon)\omega_i(t)\prod_{j\in{\cal A}(t)\setminus{\mathcal E}}[1-(1-\epsilon)\omega_j(t)] W_{t+1}(\Omega'_l(t+1))
\end{cases},
\end{eqnarray*}
where $\Omega_l(t+1)$ and $\Omega_l'(t+1)$ are the belief vector for slot $t+1$ generated by $\Omega_A(t)$ and $\Omega_{A'}(t)$ based on the belief update equation~\eqref{eq:belief_update_err_asym}.

We distinguish the following four cases:
\begin{itemize}
\item If channel $l$ is never chosen for $\Omega_{l}(t+1)$ and $\Omega'_{l}(t+1)$ from the slot $t+1$ to the end of time horizon of interest $T$, that is to say, $l\notin \mathcal{A'}(r)$ and $l\notin \mathcal{A}(r)$ for $t+1\le r \le T$,  it is easy to know $\Gamma(\Omega_{A'}(t))-\Gamma(\Omega_{A}(t)) = 0$, furthermore $W_t(\Omega_{A'}(t))-W_t(\Omega_{A}(t)) = 0$;
\item There exists $t^0$ ($t+1\leq t^0 \leq T$) such that $l\notin \mathcal{A'}(r)$ and $l\notin \mathcal{A}(r)$ for $t+1\le r \le t^0-1$ while $l\notin \mathcal{A'}(t^0)$ and $l\in \mathcal{A}(t^0)$. For this case, it holds $\mathcal{A'}(r)=\mathcal{A}(r)$ for $t+1\le r \le t^0-1$ while $\mathcal{A'}(r)$ and $\mathcal{A}(r)$ differ in one element, assume that $m\in\mathcal{A'}(t^0)$ and $m\notin \mathcal{A}(r)$. According to the definition of the myopic policy, it follows $\omega_l(t^0)\geq \omega_m(t^0)$ and $\omega'_l(t^0) \leq \omega_m(t^0)$, which leads to contradiction since $\omega'_l(t+1)=p^l_{11}>\omega_l(t+1)=p^l_{01}$ leads to $\omega'_l(t^0)>\omega_l(t^0)$ following Lemma~\ref{lemma:property_phi_err_sym}. This case is thus impossible to happen;
\item There exists $t^0$ ($t+1\leq t^0 \leq T$) such that $l\notin \mathcal{A'}(r)$ and $l\notin \mathcal{A}(r)$ for $t+1\le r \le t^0-1$ while $l\in \mathcal{A'}(t^0)$ and $l\in \mathcal{A}(t^0)$. For this case, according to the hypothesis ($l\in{{\cal A}'}$ and $l\in{\cal A}$), we have
    \begin{multline*}
    0\le W_{t^0}(\Omega_l'(t^o))-W_{t^0}(\Omega_l(t^o)) \le c\cdot(\omega_l'(t^o)-\omega_l(t^o))g'_{max}\Delta_{max}\sum_{i=0}^{T-t^o}\beta^i(\delta_p^{max})^i \\
    = c\cdot(p_{11}^l-p_{01}^l)^{t^o-t}(\omega_l'(t)-\omega_l(t))g'_{max}\Delta_{max}\sum_{i=0}^{T-t^o}\beta^i(\delta_p^{max})^i.
    \end{multline*}
    Noticing $t^0\geq t+1$, we have
    \begin{equation*}
    0  \leq W_{t+1}(\Omega_{l}'(t+1))-W_{t+1}(\Omega_{l}(t+1))\leq c\cdot(p_{11}^l-p_{01}^l)(\omega_l'(t)-\omega_l(t))g'_{max}\Delta_{max}\sum_{i=0}^{T-t-1}\beta^i(\delta_p^{max})^i.
    \end{equation*}
    Furthermore,
    \begin{align*}
    0 \leq W_{t}(\Omega_{A'}(t))-W_{t}(\Omega_{A}(t)) &=\beta(\Gamma(\Omega_{A'}(t))-\Gamma(\Omega_{A}(t))) \\
     &\leq \beta \cdot c\cdot (p_{11}^l-p_{01}^l)(\omega_l'(t)-\omega_l(t))g'_{max}\Delta_{max}\sum_{i=0}^{T-t-1}\beta^i(\delta_p^{max})^i \\ &\leq  \beta c \delta_p^{max}(\omega_l'(t)-\omega_l(t))g'_{max}\Delta_{max}\sum_{i=0}^{T-t-1}\beta^i(\delta_p^{max})^i \\
    &= c(\omega_l'(t)-\omega_l(t))g'_{max}\Delta_{max}\sum_{i=1}^{T-t}\beta^i(\delta_p^{max})^i.
    \end{align*}
\item There exists $t^0$ ($t+1\leq t^0 \leq T$) such that $l\notin \mathcal{A'}(r)$ and $l\notin \mathcal{A}(r)$ for $t+1\le r \le t^0-1$ while $l\in \mathcal{A'}(t^0)$ and $l\notin \mathcal{A}(t^0)$. For this case, by the induction hypothesis ($l\in{{\cal A}'}$ and $l\notin{\cal A}$), we have
    \begin{multline*}
    0 \leq W_{t^0}(\Omega_l'(t^o))-W_{t^0}(\Omega_l(t^o))
    \leq c\cdot(\omega_l'(t^o)-\omega_l(t^o))g'_{max}\Delta_{max}\sum_{i=0}^{T-t^o}\beta^i(\delta_p^{max})^i \\
    = c\cdot(p_{11}^l-p_{01}^l)^{t^o-t}(\omega_l'(t)-\omega_l(t))g'_{max}\Delta_{max}\sum_{i=0}^{T-t^o}\beta^i(\delta_p^{max})^i.
    \end{multline*}
    Noticing that $t+1 \leq t^o$, we have
    \begin{multline*}
    0 \leq W_{t+1}(\Omega_{l}'(t+1))-W_{t+1}(\Omega_{l}(t+1))
     \leq c\cdot(\omega_l'(t)-\omega_l(t))(p_{11}^l-p_{01}^l)g'_{max}\Delta_{max}\sum_{i=0}^{T-t-1}\beta^i(\delta_p^{max})^i.
    \end{multline*}
    Therefore, we have
    \begin{align*}
0 &\leq W_{t}(\Omega_{A'}(t))-W_{t}(\Omega_{A}(t)) \\
& \leq \beta \cdot c\cdot(\omega_l'(t)-\omega_l(t))(p_{11}^l-p_{01}^l)g'_{max}\Delta_{max}\sum_{i=0}^{T-t-1}\beta^i(\delta_p^{max})^i
    \\
    &\leq \beta\cdot c\cdot(\omega_l'(t)-\omega_l(t))\delta_p^{max}g'_{max}\Delta_{max}\sum_{i=0}^{T-t-1}\beta^i(\delta_p^{max})^i \\
     &= c(\omega_l'(t)-\omega_l(t))g'_{max}\Delta_{max}\sum_{i=1}^{T-t}\beta^i(\delta_p^{max})^i.
    \end{align*}
\end{itemize}

Combining the above results, we complete the proof of the second part ($l\notin{{\cal A}'}$ and $l\notin{\cal A}$) of Lemma~\ref{lemma:symmetry_asym_err}.

Last, \textbf{we prove the third case $l \in {\cal{A}}'(t)$ and $l\notin{{\cal A}(t)}$}. In this case, there must exist a channel $m$ such that $\omega'_l \geq \omega_{m} \geq \omega_l$ and $\omega'_l\in\mathcal{A'}$ and $\omega_m\in\mathcal{A}$. We then have
\begin{align}
 & W_{t}(\Omega_{A'}(t))-W_{t}(\Omega_{A}(t)) \nonumber \\
  =& W_{t}(\omega_1,\cdots,\omega'_l,\cdots,\omega_N) - W_{t}(\omega_1,\cdots,\omega_l,\cdots,\omega_N) \nonumber \\
  =& W_{t}(\omega_1,\cdots,\omega'_l,\cdots,\omega_N) - W_{t}(\omega_1,\cdots,\omega_l=\omega_m,\cdots,\omega_N) \nonumber \\
  & +W_{t}(\omega_1,\cdots,\omega_l=\omega_m,\cdots,\omega_N)-W_{t}(\omega_1,\cdots,\omega_l,\cdots,\omega_N)
\label{eq:case3}
\end{align}

According to the induction hypothesis ($l\in{{\cal A}'}$ and $l\in{\cal A}$), the first term of the right hand of~\eqref{eq:case3} can be bounded as follows:
\begin{align}
  0 &\leq W_{t}(\omega_1,\cdots,\omega'_l,\cdots,\omega_N) - W_{t}(\omega_1,\cdots,\omega_l=\omega_m,\cdots,\omega_N) \nonumber\\
    &\leq c\cdot(\omega_l'(t)-\omega_{m}(t))g'_{max}\Delta_{max}\sum_{i=0}^{T-t}\beta^i(\delta_p^{max})^i
\label{eq:case3_1}
\end{align}

Meanwhile, the second term of the right hand of~\eqref{eq:case3} is bounded by induction hypothesis ($l\notin{{\cal A}'}$ and $l\notin{\cal A}$) as:
\begin{align}
\label{eq:case3_2}
   0 &\leq W_{t}(\omega_1,\cdots,\omega_l=\omega_m,\cdots,\omega_N)-W_{t}(\omega_1,\cdots,\omega_l,\cdots,\omega_N) \nonumber \\
    &\leq  c\cdot(\omega_{m}(t)-\omega_l(t))g'_{max}\Delta_{max}\sum_{i=1}^{T-t}\beta^i(\delta_p^{max})^i
\end{align}
Therefore, we have, combining~\eqref{eq:case3},~\eqref{eq:case3_1} and ~\eqref{eq:case3_2},
\begin{equation*}
0 \leq W_{t}(\Omega_{A'}(t))-W_{t}(\Omega_{A}(t)) \leq c\cdot(\omega_l'(t)-\omega_l(t))g'_{max}\Delta_{max}\sum_{i=0}^{T-t}\beta^i(\delta_p^{max})^i,
\end{equation*}
which completes the proof of the third part ($l\in{{\cal A}'}$ and $l\notin{\cal A}$) of Lemma~\ref{lemma:symmetry_asym_err}. Lemma~\ref{lemma:symmetry_asym_err} is thus proven.

\bibliographystyle{unsrt}
\bibliography{reference}

\begin{thebibliography}{1}

\bibitem{Papadimitriou99}
C.~H. Papadimitriou and J.~N. Tsitsiklis.
\newblock The complexity of optimal queueing network control.
\newblock {\em Mathematics of Operations Research}, 24(2):293--305, 1999.

\bibitem{Qzhao08}
{Q. Zhao, and B. Krishnamachari, and K. Liu}.
\newblock On myopic sensing for multi-channel opportunistic access: Structure,
  optimality, and performance.
\newblock {\em IEEE Transactions Wireless Communication}, 7(3):5431--5440, Dec.
  2008.

\bibitem{Ahmad09b}
{S. Ahmad, M. Liu, T. Javidi and Q. Zhao and B. Krishnamachari}.
\newblock {Optimality of Myopic Sensing in Multi-Channel Opportunistic Access}.
\newblock {\em IEEE Transactions on Information Theory}, 55(9):4040--4050,
  2009.

\bibitem{Ahmad09}
S.~Ahmad and M.~Liu.
\newblock Multi-channel opportunistic access: a case of restless bandits with
  multiple plays.
\newblock In {\em Allerton Conference}, Monticello, Il, Spet.-Oct. 2009.

\bibitem{Wang11TPS}
K.~Wang and L.~Chen.
\newblock On optimality of myopic policy for restless multi-armed bandit
  problem: An axiomatic approach.
\newblock {\em IEEE Transactions on Signal Processing}, 60(1):300--309, 2012.

\bibitem{Kliu10}
{K. Liu, and Q. Zhao, and B. Krishnamachari}.
\newblock Dynamic multichannel access with imperfect channel state detection.
\newblock {\em IEEE Transactions on Signal Processing}, 58(5):2795--2807, May
  2010.

\bibitem{Wang11TCOM}
{K. Wang, L. Chen, Q. Liu and Khaldoun Al Agha}.
\newblock On optimality of myopic sensing policy with imperfect sensing in
  multi-channel opportunistic access.
\newblock {\em Computing Research Repository (CoRR) arXiv:1202.0477}, 2011.

\end{thebibliography}

\end{document}